\documentclass[conference]{IEEEtran}
\IEEEoverridecommandlockouts
\usepackage{cite}
\usepackage{amsmath,amssymb,amsfonts,amsthm}
\usepackage{algorithmic}
\usepackage{graphicx}
\usepackage{textcomp}
\usepackage{xcolor}
\usepackage{ytableau}
\def\BibTeX{{\rm B\kern-.05em{\sc i\kern-.025em b}\kern-.08em
    T\kern-.1667em\lower.7ex\hbox{E}\kern-.125emX}}
    
\newtheorem{theorem}{Theorem}[section]
\newtheorem{lemma}[theorem]{Lemma}

\newtheorem{corollary}[theorem]{Corollary}
\newtheorem{definition}[theorem]{Definition}

\newtheorem{remark}[theorem]{Remark}

\def\t{\rm t}

\def\id{\rm id}

\def\t{\rm t}

\def\bbC{\mathbb C}
\def\bbG{\mathbb G}

\def\bb1{\mathds 1}

\newcommand{\bbS}{\mathbb{S}}

\newcommand{\la}{\langle}
\newcommand{\ra}{\rangle}

\def\U{\mathcal U}

\def\E{\mathcal E}
\begin{document}

\title{Frames for Graph Signals on the Symmetric Group:\\
A Representation Theoretic Approach
\thanks{This project was funded by University of Delaware (UNIDEL summer fellowship). The second author was supported by NSF grant DMS-1902301.}}

\author{\IEEEauthorblockN{Kathryn Beck}
\IEEEauthorblockA{\textit{Department of Mathematical Sciences} \\
\textit{University of Delaware}\\
Newark, DE, USA \\
kbeck@udel.edu}
\and
\IEEEauthorblockN{Mahya Ghandehari}
\IEEEauthorblockA{\textit{Department of Mathematical Sciences} \\
\textit{University of Delaware}\\
Newark, DE, USA \\
mahya@udel.edu}
% \and
% \and
% \and
% \and
}

\maketitle

\begin{abstract}
An important problem in the field of graph signal processing is developing appropriate overcomplete dictionaries for signals defined on different families of graphs. The Cayley graph of the symmetric group has natural applications in ranked data analysis, as its vertices represent permutations, while the generating set formalizes a notion of distance between rankings. Taking advantage of the rich theory of representations of the symmetric group, 
we study a particular class of frames, called Frobenius--Schur frames, where every atom belongs to the coefficient space of only one irreducible representation of the symmetric group. 
We provide a characterization for all  Frobenius--Schur frames on the group algebra of the symmetric group which are ``compatible''  with respect to the generating set.
Such frames have been previously studied for the permutahedron, the Cayley graph of the symmetric group with the generating set of adjacent transpositions, and have proved to be capable of producing meaningful interpretation of the ranked data set via the analysis coefficients. 
Our results generalize frame constructions for the permutahedron  to any inverse-closed generating set.

\end{abstract}

\begin{IEEEkeywords}
Graph frame, Cayley graph, permutation group
\end{IEEEkeywords}

%%%%%%%%%Intro%%%%%%%%%%%%%%%%%%%%%%%%%%%%%%%%%%%%%%%%%%%%%%%%%%%%%%%%%%%%%%%%%%%%%%%%%%%%%%%%%%%%%%%%%%%%%
\section{Introduction}
A signal on a graph $G$ is a complex-valued function $f$ on the vertex set of $G$. Fixing an ordering $\{v_i\}_{i=1}^N$ of the vertex set, a graph signal $f$ can be represented as a column vector $\left[f(v_1), f(v_2),\cdots, f(v_N)\right]^{\t}$ in $\bbC^N$, where $\t$ denotes the matrix transpose operation. A major objective of the vibrant field of graph signal processing is to analyze such signals not only as vectors in $\bbC^N$, but to take the underlying structure of the graph $G$ into account. 
Over the past few years, the problem of generalizing or adapting classical tools of Fourier analysis to the context of graph signals has attracted the attention of many researchers. 
For a more detailed introduction to graph signal processing and its applications, see \cite{2018:Ortega:GraphSignalOverview}, \cite{2013:Shuman:EmergingFieldSignal}, and \cite{2019:Stankovic:IntroGSP}.

An important technique for analyzing signals in general, and graph signals in particular, is to develop appropriate overcomplete dictionaries for various classes of signals. This idea is formalized in the theory of discrete frames. 
A \emph{frame} for a  finite-dimensional (or infinite-dimensional separable) Hilbert space $\mathcal H$ is a set of vectors $\{\psi_x\}_{x\in X}$ indexed by a countable set $X$, such that for some positive real numbers $A$ and $B$, we have for every $v\in \mathcal{H}$,
\begin{equation}\label{Eframedefn}
  A\|v\|_{\mathcal H}^2
    \leq \sum_{x\in X} |\langle v, \psi_x \rangle |^2
    \leq B\|v\|_{\mathcal H}^2.%
\end{equation}%
The constants $A$ and $B$ in \eqref{Eframedefn} are called the \emph{lower frame bound} and the \emph{upper frame bound} respectively,  
and the \emph{condition number} of the frame is defined to be the ratio $c(\mathcal{F}) := B/A$. 
Frames provide stable, possibly redundant systems which allow reconstruction of a signal $f$ from its frame coefficients $\{\langle f, \psi_x \rangle\}_{x\in X}$. 
When the frame provides a redundant representation, reconstruction of a signal is still possible even when some portion of its frame coefficients are lost or corrupted.
An important class of frames is the class of \emph{tight frames}, {\it i.e.}, frames for which $A=B$.
Compared to general frames, tight frames exhibit greater numerical stability when reconstructing noisy signals.
\emph{Parseval frames} are tight frames in which $A=B=1$.

Over the past couple of decades, various methodologies for constructing frames for graph signals have been investigated. 
In \cite{Hammond:2011:WaveletsOnGraphViaSpectralGraphTheory},
Hammond, Vandergheynst and Gribonval define the graph Fourier transform and apply it to produce wavelet frames for graphs. Other examples of wavelet-type frames can be found in \cite{2017:Dong:SparseRepWavelet,2018:Gobel:TightFramesDenoising,2013:Leonardi:TightWaveletFrames,2013:Shuman:SpectrumAdaptedWavelet}. 
Another significant class of classical frames are Gabor frames. These frames are constructed through applications of translation and modulation operators to a window function. For various constructions of  Gabor-type frames for graph signals, we refer the reader to \cite{2016:Behjat:SignalAdaptedFrames,  2021:Ghandehari:GaborTypeFrames}. 
A survey of localized spectral graph filter frames can be found in \cite{2020:Shuman:UnifyingFramework}.
In this paper, we focus our attention to a particular class of graphs, namely the Cayley graphs of the symmetric group $\bbS_n$, and we use the representation theory of finite groups to construct frames of a ``suitable type'' for signals on these graphs. 
Cayley graphs of $\bbS_n$ have natural applications in ranked data analysis.
The vertex set of such a graph represents all preference rankings of $n$ objects or candidates, as each vertex is associated with a permutation of $n$ items. The choice of the generating set for the Cayley graph formalizes the idea of distance between rankings in the context of a ranked voting model. 
An important example is the \emph{permutahedron}, denoted $\mathbb{P}_n$.
This is the Cayley graph of $\mathbb S_n$ with generating set $S$ of all adjacent transpositions $(i,i+1)$. 
In this case, two rankings are considered  ``close together'' if one ranking can be obtained from the other by switching candidates that are closer together.
In \cite{paper}, Chen et al.~construct a particular frame for signals on the permutahedron, and show that the analysis coefficients with respect to this frame are meaningful in the context of ranked data analysis. 
Examples of interpretations of the analysis coefficients include popularity of candidates, whether a candidate is polarizing, or whether two candidates are likely to be ranked similarly. In order to obtain frames which lead to meaningful analysis coefficients, they build frames in which every atom belongs to the coefficient space of one irreducible representation only. Such frames also allow for interpretation of symmetry and smoothness in the ranked data set.
Structurally speaking, these frames satisfy the favourable property of being compatible with the coefficient spaces of representations of ${\mathbb S}_n$; this motivates the following definition.
\begin{definition}
Let $n\in {\mathbb N}$. 
A frame $\{\psi_i\}_{i=1}^m$  for $\bbC[\bbS_n]$ is called a Frobenius--Schur frame if each atom $\psi_i$ belongs to one orthogonal component of the Frobenius--Schur
decomposition as stated in Theorem~\ref{thm:schur-ortho} (iii).
\end{definition}
This article is dedicated to the construction of Frobenius--Schur frames for signals on all Cayley graphs on $\mathbb S_n$, with \emph{any} choice of the inverse-closed generating set. Namely, we provide a characterization for all possible Frobenius--Schur frames on $\bbC[\bbS_n]$ which satisfy certain compatibility conditions with respect to the generating set $S$. As our method relies heavily on the concrete form of the irreducible representations of the group, we focus our attention on $\bbS_n$, given the rich literature on the concrete forms of representations of $\bbS_n$.

We believe that our results can be used for analysis of ranked data sets in a wide range of settings, 
as taking different generating sets allows us to model   ``closeness'' of ranked data in a variety of ways. Heuristically-speaking, it may be advantageous to define rankings to be closer together if one ranking can be obtained from the other by switching candidates that are near each other in, say, the top half of the ranking. This could be useful in a ranked vote where only the first place candidate is selected as the ``winner'', and the analysis is focused on candidates placed higher up in the preference rankings. As another example, the generating set of all transpositions allows for a more highly connected notion of closeness where rankings are close together if any two candidates are switched.

This paper is organized as follows. In Section~\ref{section: notations}, we collate the necessary background on Cayley graphs and the representation theory of ${\mathbb S}_n$. In Section~\ref{sec:frames}, we define the notion of frame compatibility with a generating set and present our main results for constructing Frobenius--Schur Frames for $L^2(\bbS_n)$. 
%present the main results characterizing all Frobenius--Schur Frames for $L^2(\bbS_n)$.
In Section~\ref{sec:example} we provide an explicit example of Frobenius--Schur Frames for $L^2(\bbS_3)$ that are compatible with the generating set of adjacent transpositions.

%%%%%%%%%%%%%%%%%%%%%%%%%%%%%%%%%%%%%%%%%%%%%%%%%%%%%%%%%%%%%%%%%%%%%%%%%%%%%%%%%%%%%%%%%%%%%%%%%%%%%%%%%%%%%%%%%%%%%%%%%%%%%%%%%%%%%%%%%%%%%%%%%%%%%%%%
\section{Notations and Background}\label{section: notations}
Throughout this article, we use $\bbG$ to denote a finite (not necessarily Abelian)  group of size $N$. The space of all signals $f:\bbG\rightarrow \bbC$ is denoted by $\bbC[\bbG]$. The group algebra $\bbC[\bbG]$, equipped with inner product
$\langle f,g\rangle=\sum_{x\in\bbG} f(x)\overline{g(x)}$, is a Hilbert space isometrically isomorphic to $\bbC^{N}$, which we denote by $L^2(\bbG)$.
Given that a signal $f:\bbG \rightarrow \bbC$ in $\mathbb C[\mathbb G]$ can be viewed as a vector in the vector space $\mathbb C^{N}$, 
for the remainder of the paper, we refer to $\mathbb C[\mathbb G]$ and $L^2(\mathbb G)$ interchangeably, where we view $f$ as a function in the context of $\mathbb C[\mathbb G]$ and as a vector in the context of $L^2(\mathbb G)$.

Let ${\mathbb G}$ be a finite group, and $S\subseteq {\mathbb G}$ be an inverse-closed subset of ${\mathbb G}$ (i.e., if $x\in S$ then $x^{-1}\in S$). The \textit{Cayley graph} $G({\mathbb G},S)$ encodes the group structure of ${\mathbb G}$ with respect to $S$. Namely, the vertex set of $G({\mathbb G},S)$ is ${\mathbb G}$, and vertices $x,y\in{\mathbb G}$ form an edge if  $x^{-1}y\in S$. The set $S$ is called the \emph{generating set} of the Cayley graph $G({\mathbb G},S)$. The fact that the generating set is inverse-closed guarantees that the associated Cayley graph is not directed. 

In the next subsection, we provide the necessary background for the representation theory of finite groups and their associated function spaces in general. We then focus our attention to these concepts for ${\mathbb S}_n$. 
\subsection{The Frobenius--Schur Decomposition}\label{subsec:decomp}
%reword and possibly condense
%
A \emph{unitary representation} of $\bbG$ of dimension $d$ is a group homomorphism $\pi: \bbG\rightarrow \U_d(\bbC)$, where $\U_d(\bbC)$ denotes the (multiplicative) group of unitary matrices of size $d$. 
Here, we restrict our attention to unitary representations. This is non-consequential as every representation of a finite group can be turned into a unitary representation by a change of inner product on the representation space (see for example  \cite[Section 1.3]{Serre}). 
For a given (unitary) representation $\pi$,  a subspace $W$ of $\bbC^d$ is called \emph{$\pi$-invariant} if $\pi(g)W:=\{\pi(g)\xi: \ \xi\in W\}\subseteq W$ for all $g\in \mathbb{G}$. A representation $\pi$ is called \emph{irreducible} if $\{0\}$ and $\mathbb{C}^d$ are its only closed $\pi$-invariant subspaces. 
Every unitary representation of a finite group completely decomposes into a direct sum of its irreducible representations.
Two representations $\pi$ and $\sigma$ of $\bbG$ are called \emph{unitarily equivalent} if there exists a unitary matrix $U$ such that $U^{-1}\pi(g)U=\sigma(g)$ for all $g\in \bbG$. 
We let $\widehat{\bbG}$ denote the collection of all (equivalence classes of) irreducible unitary representations of $\bbG$. 

For an arbitrary $\pi\in\widehat{\bbG}$ of dimension $d_\pi$, and vectors $\xi,\eta\in\bbC^{d_\pi}$, we define the \emph{coefficient function} associated with the representation $\pi$ and the vectors $\xi, \eta$ as follows:
\begin{equation}\label{equation:coefficient function}
\pi_{\xi,\eta}:\bbG\to \bbC, \ \ \pi_{\xi,\eta}(g)=\langle\pi(g)\xi, \eta\rangle, \ \forall g\in \bbG.
\end{equation}
When the group $\bbG$ is ordered as $\bbG=\{g_1,\ldots, g_N\}$, we can represent $\pi_{\xi,\eta}$ as a vector in $\bbC^N$, namely, 
$$\pi_{\xi,\eta}=\Big[\langle\pi(g_1)\xi, \eta\rangle, \ldots, \langle\pi(g_N)\xi, \eta\rangle\Big]^{\t}.$$
Given a standard orthonormal basis $\{e_i\}_{i=1}^{d_\pi}$ for $\bbC^{d_\pi}$, the coefficient functions 
$$\pi_{i,j}(x):=\pi_{e_i, e_j}(x)=\la \pi(x)e_i,e_j \ra, \ i,j=1,\ldots, d_\pi$$
are the entries of the matrix of $\pi(x)$ represented in the same basis. 

Coefficient functions, and the subspaces generated by them, play a central role in the harmonic analysis of non-Abelian groups. 
Let $\pi\in \widehat{\mathbb G}$ and $1\leq i\leq d_{\pi}$ be fixed. Define
\begin{equation}\label{equation:E-pi-i}
    {\mathcal E}_{\pi,i}=\big\{\pi_{\xi,e_i}:\xi\in \bbC^{d_\pi}\big\},
\end{equation}
the space of all coefficient functions of $\pi$, fixed in the second entry. It is easy to observe that for each $i$, the set  ${\mathcal E}_{\pi,i}$ forms a right-invariant subspace of $\bbC^N$, that is, for every $f\in{\mathcal E}_{\pi,i}$ and $y\in \bbG$, the map $f_y:\bbG\to \bbC$ defined as $f_y(x)=f(xy)$ also belongs to ${\mathcal E}_{\pi,i}$.
The well-known theorem of Frobenius and Schur provides a direct sum decomposition of $L^2(\mathbb G)$ into subspaces of the form  ${\mathcal E}_{\pi,i}$. 
\begin{theorem}[Frobenius--Schur decomposition]\label{thm:schur-ortho}
Let $\bbG$ be a finite group, and $\pi,\sigma$ be irreducible unitary representations of $\bbG$.
\begin{enumerate}
    \item[(i)]\label{thm:schur1} If $\pi$ and $\sigma$ are not unitarily equivalent then $\E_{\pi,i} \perp \E_{\sigma,j}$ for all $1\leq i\leq d_\pi$ and $1\leq j\leq d_\sigma$.
    \item[(ii)] Every orthonormal basis  $\{e_j\}_{j=1}^{d_\pi}$ for $\bbC^{d_\pi}$ leads to an orthonormal basis for $\E_{\pi,i}$ given by
    \begin{equation}\label{equation:schur orthogonality}
    \left\{\sqrt{\frac{d_\pi}{|\bbG|}} \pi_{j,i}:j=1,..,d_\pi \right\}.
    \end{equation}
    \item[(iii)]\label{thm:schur2} $L^2(\bbG)=\bigoplus_{\pi \in \widehat{\bbG}}\bigoplus_{1\leq i\leq d_\pi} {\cal E}_{\pi,i}$, with the associated orthonormal basis
    \[ \left\{\phi^{\pi}_{i,j}:=\sqrt{\frac{d_{\pi}}{|\bbG|}} \pi_{i,j}: i,j=1,...,d_\pi,\  \pi \in \widehat{\bbG} \right\}.  \]
    We call ${\cal E}_{\pi,i}$ the orthogonal subspaces of the Frobenius--Schur decomposition.
\end{enumerate}
\end{theorem}
\subsection{Representations of $\mathbb S_n$}\label{subsec:rep}
In this section, we present the foundational tools for describing irreducible representations of $\bbS_n$, focusing on only the background necessary to our construction of Frobenius--Schur frames for $L^2(\bbS_3)$ in Section \ref{sec:example}. 
%As the representations of $\bbS_n$ have been well-documented by many, we choose to 
%
Irreducible representations of ${\mathbb S}_n$ are in one-to-one correspondence with the partitions $\lambda \vdash n$ of $n$. A \emph{partition} $\lambda$ of the integer $n$, denoted $\lambda \vdash n$, is a decomposition of $n$ into a sum of positive integers. Equivalently,  $\lambda: (\lambda_1,..,\lambda_k)$ is a partition of $n$ if $\lambda_1\geq \ldots\geq \lambda_k$ and $n=\sum_{i=1}^k \lambda_i$.
A partition $\lambda: (\lambda_1,..,\lambda_k)$ can be represented by a Young diagram with shape $\lambda : (\lambda_1,..,\lambda_k)$, that is, a Young diagram that contains $\lambda_i$ boxes in its $i$th row, for every $1\leq i\leq k$. As a result, every Young diagram of size $n$ corresponds with an irreducible representation of $\bbS_n$. 
A further extension of the Young diagram is the Young tableau, a Young diagram on $n$ blocks where each block is uniquely labeled from the set $\{1,2,\dots, n-1,n\}$. We say that a Young tableau is in standard form if the labels in each row increase from left to right and the labels in each column increase from top to bottom. 
%
%For instance, the following are all of the standard Young tableaux corresponding with the partition $(3,1)$:
%
Young tableaux provide significant concrete information about the irreducible representations of $\bbS_n$. For example, the number of standard Young tableaux for a given partition $\lambda \vdash n$ is the dimension of the associated irreducible representation of $\bbS_n$. We refer the reader to \cite{FultonHarris} for more details on the information that can be retrieved from Young tableaux, such as formulas regarding the characters and the concrete matrix form of the representations.

%%%%%%%%%%%%%%%%%%%%%%%%%%%%%%%%%%%%%%%%%%%%%%%%%%%%%%%%%%%%%%%%%%%%%%%%%%%%%%%%%%%%%%%%%%%%%%%%%%%%%%%%%%%%%%%%%%%%%%%%%%%%%%%%%%%%%%%%%%%%%%%%%%%%%%%%%%%%%%%%%%%%%%%%%%%%%%%%%%%%%%%%%%%%%%%%%%%%%%%%%%%%%%%%%%%%%%

\section{Frobenius--Schur frames for $\mathbb S_n$}\label{sec:frames}
Throughout this section, let $G$ be the Cayley graph of $\bbS_n$ with a fixed generating set $S\subseteq \bbS_n$. We  construct Frobenius--Schur frames for $L^2(\bbS_n)$ that are ``compatible'' with the generating set $S$. 
For a representation $\pi$ of $\bbS_n$, we define $\pi(S)=\sum_{a\in S}\pi(a)$. Since $S$ is inverse-closed and $\pi$ is unitary, the matrix $\pi(S)$ is self-adjoint.
For the remainder of the paper, we fix an ordering of the elements of $\bbS_n$, and identify elements of $L^2(\bbS_n)$ with vectors in $\bbC^{|\bbS_n|}$, when needed.
\begin{definition}\label{def:frame-compatible}
A Frobenius--Schur frame $\{\phi_i\}_{i=1}^m$ of $L^2(\bbS_n)$ is said to be \emph{compatible with $S$} if for every atom $\phi_i$ 
there exists an irreducible representation $\pi:\bbS_n\to {\mathcal U}_{d_\pi}(\bbC)$ and an index $1\leq j\leq d_\pi$ such that,
\begin{equation}
    \phi_i(z)=R_{\pi,j}(z)X,\ z\in \bbS_n
\end{equation}
where $R_{\pi,j}(z)$ is the $j$th row of $\pi(z)$, and $X$ is an eigenvector of $\pi(S)$. 
\end{definition}

Let $\pi$ be an irreducible representation of ${\bbS_n}$ of dimension $d_\pi$, and let $1\leq i\leq d_\pi$ be fixed. 
For $\lambda\in {\mathbb R}$, define 
\begin{equation}\label{eq:Z-space}
Z_{\pi, i, \lambda}:=\left\{\sum_{k=1}^{d_\pi}x_k\pi_{k,i}: \ 
\left[
\begin{array}{c}
x_1 \\
\vdots \\
x_{d_\pi}  
\end{array}
\right]
\in E_{\lambda}(\pi(S)) \right\},
\end{equation}
where $E_{\lambda}(\pi(S))$ is the $\lambda$-eigenspace of $\pi(S)$. Since $\pi_{k,i}$ belongs to $\E_{\pi,i}$ for every $k=1,\ldots, d_\pi$, 
$Z_{\pi, i, \lambda}\subseteq \E_{\pi,i}$.
Also, note that if $\lambda$ is not an eigenvalue of $\pi(S)$, then $Z_{\pi, i, \lambda}=\{{ 0}\}$. 

%The following simple lemma relates Definition~\ref{def:frame-compatible} with the sets $Z_{\pi, i, \lambda}$.
\begin{lemma}\label{lem:frob-schur-Z description}
A frame for $L^2(\bbS_n)$ is a Frobenius--Schur frame compatible with $S$ \emph{iff} every frame atom belongs to $Z_{\pi, i, \lambda}$ for some $\pi\in \widehat{\bbS_n}$, some real scalar $\lambda$, and some index $i$.
\end{lemma}
\begin{proof}
The forward direction follows easily from direct calculations. 
For the converse, suppose that for a given frame, every atom belongs to $Z_{\pi, i, \lambda}$ for some $\pi\in \widehat{\bbS_n}$, some real scalar $\lambda$, and some index $i$. 
Since $Z_{\pi, i, \lambda}\subseteq \E_{\pi,i}$, the frame must be a Frobenius--Schur frame. The fact that the frame is compatible with $S$ trivially follows from the definition of $Z_{\pi, i, \lambda}$.
\end{proof}

The sets $Z_{\pi, i, \lambda}$ are isomorphic to the $\lambda$-eigenspace of $\pi(S)$.
\begin{lemma}\label{lemma: isometry}
With notation as in \eqref{eq:Z-space}, we have:
\begin{itemize}
\item[(i)] The set $Z_{\pi, i, \lambda}$ is a subspace of $\E_{\pi,i}\subseteq L^2(\bbS_n)$.
\item[(ii)] $Z_{\pi, i, \lambda}$ and $E_{\lambda}(\pi(S))$ are isomorphic as linear spaces via the map $\Theta_{\pi, i, \lambda}:E_{\lambda}(\pi(S))\to Z_{\pi, i, \lambda},$ defined by
\begin{equation}\label{eq:isom-theta}
\Theta_{\pi, i, \lambda}([x_1 \ldots x_{d_{\pi}}]^{\t})= \sum_{k=1}^{d_\pi}x_k\pi_{k,i}.
\end{equation}
So, the subspace $Z_{\pi, i, \lambda}$ is the trivial subspace containing only ${0}\in L^2({\mathbb S}_n)$ \emph{iff} $\lambda$ is not an eigenvalue of $\pi(S)$.
\item[(iii)] The map $\widetilde{\Theta}_{\pi, i, \lambda}:=\sqrt{\frac{d_{\pi}}{n!}}\Theta_{\pi, i, \lambda}$ is an isometry. 
\item[(iv)] If $(\pi, i, \lambda)\neq (\pi', i', \lambda')$, then $Z_{\pi, i, \lambda}$ and $Z_{\pi', i', \lambda'}$ form orthogonal subspaces in $L^2(\bbS_n)$.
\end{itemize}
\end{lemma}

\begin{proof}
Part (i) holds, since $E_{\lambda}(\pi(S))$ is a linear subspace of ${\mathbb C}^{d_{\pi}}$. 
To prove part (ii), note that $\Theta_{\pi, i, \lambda}$ is clearly a linear surjective map. To show injectivity, assume $\Theta_{\pi, i, \lambda}([x_1 \ldots x_{d_{\pi}}]^{\t})=\sum_{k=1}^{d_\pi}x_k\pi_{k,i}=0$. Since $\{\pi_{k,i}: 1\leq k\leq d_\pi\}$ is a linearly independent subset of $L^2({\mathbb S}_n)$, we 
get $x_k=0$ for all $k$, and $\Theta_{\pi,i,\lambda}$ is injective.
Thus,  $\Theta_{\pi, i, \lambda}$ is a linear isomorphism. %between the linear spaces $Z_{\pi,i,\lambda}$ and $E_{\lambda}(\pi(S))$.
As a consequence, $\Theta_{\pi, i, \lambda}$ keeps the dimension unchanged; this implies the last statement. 

To prove (iii), let $X=[x_1\dots x_{d_{\pi}}]^{\t}, Y=[y_1\dots y_{d_{\pi}}]^{\t}$ be arbitrary elements of $E_{\lambda}(\pi(S))$. 
By  the Schur orthogonality relations (Theorem~\ref{thm:schur-ortho}), we have
\begin{align*}
    &\left\langle \widetilde{\Theta}_{\pi,i,\lambda}(X),\widetilde{\Theta}_{\pi,i,\lambda}(Y)\right\rangle%_{L^2({\mathbb S}_n)}
    =\frac{d_{\pi}}{n!} \sum_{k,\ell=1}^{d_{\pi}} x_k \overline{y_{\ell}}\left\langle \pi_{k,i}, \pi_{\ell,i}\right\rangle_{L^2({\mathbb S}_n)}\\
    &=\frac{d_{\pi}}{n!} \sum_{k=1}^{d_{\pi}} x_k\overline{y_k} \frac{n!}{d_{\pi}}
    =\sum_{k=1}^{d_{\pi}}x_k \overline{y_k}=\left\langle X,Y \right\rangle_{{\mathbb C}^{d_\pi}}.
\end{align*}

To prove part (iv), suppose $(\pi, i, \lambda)\neq (\pi', i', \lambda')$. Since $Z_{\pi, i, \lambda}\subseteq \E_{\pi,i}$ and 
$Z_{\pi', i', \lambda'}\subseteq \E_{\pi',i'}$, by Theorem~\ref{thm:schur-ortho}, if $\pi\neq \pi'$ or $i\neq i'$, the two spaces $Z_{\pi, i, \lambda}$ and $Z_{\pi', i', \lambda'}$ are orthogonal. Now suppose that $\pi=\pi'$ and $i=i'$, but $\lambda\neq \lambda'$. We have that $E_\lambda(\pi(S))$ and $E_{\lambda'}(\pi(S))$ are orthogonal subspaces of $\bbC^{d_\pi}$. So for $X\in E_{\lambda}(\pi(S))$ and $Y\in E_{\lambda'}(\pi(S))$, we have 
\begin{align*}
    &\left\langle \sum_{k=1}^{d_\pi}x_k\pi_{k,i} , \sum_{\ell=1}^{d_\pi}y_\ell\pi_{\ell,i}\right\rangle_{L^2({\mathbb S}_n)}=\sum_{k,\ell=1}^{d_\pi} x_k\overline{y_\ell}\langle\pi_{k,i} ,\pi_{\ell,i}\rangle\\
    &=\sum_{k=1}^{d_\pi} x_k\overline{y_k}\frac{d_\pi}{n!}=\frac{d_\pi}{n!}\langle X,Y\rangle_{{\mathbb C}^{d_\pi}}=0.
\end{align*}
\end{proof}
\begin{theorem}\label{thm:Z-pi-decomp}
For every $\pi\in\widehat{\mathbb S}_n$ and $1\leq i\leq d_\pi$, we have $\E_{\pi,i}=\bigoplus_{\lambda\in\sigma(\pi(S))} Z_{\pi,i,\lambda}.$
\end{theorem}
\begin{proof}
The matrix $\pi(S)$ is  self-adjoint (i.e., $\pi(S)^*=\pi(S)$), and as a result it is diagonalizable. So, one can build an orthonormal basis of $\bbC^{d_\pi}$ consisting of eigenvectors of $\pi(S)$. In other words, we can write 
\begin{equation}\label{eq:spectral-decomp}
    \bbC^{d_\pi}=\bigoplus_{\lambda\in\sigma(\pi(S))} E_\lambda(\pi(S)), 
\end{equation}
where $\sigma(\pi(S))$ denotes the spectrum of the matrix $\pi(S)$.
Now, consider an arbitrary element $\pi_{\xi,i}\in\E_{\pi,i}$, and write its linear expansion $\pi_{\xi,i}=\sum_{k=1}^{d_\pi} x_k\pi_{k,i}$. Using \eqref{eq:spectral-decomp}, the vector 
$X=[x_1,\ldots,x_{d_\pi}]^{\t}\in \bbC^{d_\pi}$ can be written as a linear combination $X=\sum_{\lambda\in \sigma(\pi(S))} Y_\lambda$, with $Y_\lambda\in E_\lambda(\pi(S))$. Denoting the $i$th row of $\pi(z)$ by $R_{\pi,i}(z)$, for every $z\in\bbS_n$, we have
$\pi_{\xi,i}(z)=(R_{\pi,i}(z)) X=\sum_{\lambda\in \sigma(\pi(S))} R_{\pi,i}(z)Y_\lambda. $
  %  =\sum_{\lambda\in \sigma(\pi(S))} \big(\sum_{k=1}^{d_\pi} (y_\lambda)_k\pi_{k,i}(z)\big),
%
Let $(y_\lambda)_k$ denote the $k$th component of $Y_\lambda$. For each $\lambda$, we have $R_{\pi,i}(z)Y_\lambda=\sum_{k=1}^{d_\pi} (y_\lambda)_k\pi_{k,i}$ belongs to $Z_{\pi,i,\lambda}$. So,  
we conclude that 
$$\E_{\pi,i}\subseteq \sum_{\lambda\in \sigma(\pi(S))} Z_{\pi,i,\lambda}.$$ 
On the other hand,  $\E_{\pi,i}\supseteq \sum_{\lambda\in \sigma(\pi(S))} Z_{\pi,i,\lambda}$ is a trivial consequence of the definition of $Z_{\pi, i, \lambda}$. 
Finally, by Lemma~\ref{lemma: isometry} (iv), this sum is a direct sum.
\end{proof}

%Before we state our main theorem, we need to introduce one last notation.
\begin{definition}
For $\pi\in \widehat{\bbS_n}$ and an eigenvalue $\lambda$ of $\pi(S)$, let ${\mathcal G}_{\pi,\lambda}$ denote the collection of all frames for  $E_\lambda(\pi(S))$. We define ${\mathcal G}_{\pi,\lambda}=\emptyset$ if $\lambda$ is not an eigenvalue of $\pi(S)$.
Elements of ${\mathcal G}_{\pi,\lambda}$ are denoted by calligraphic font, e.g.~${\cal F}$.
\end{definition}

\begin{theorem}\label{thm: frame} 
For every $\pi\in \widehat{\bbS_n}$ of dimension $d_\pi$ and every eigenvalue $\lambda$ of $\pi(S)$, let 
${\cal F}^{\pi,\lambda}_1\ldots, {\cal F}^{\pi,\lambda}_{d_\pi}\in {\mathcal G}_{\pi,\lambda}$ be given frames. 
Then we have the following.
\begin{itemize}
\item[(i)] The collection 
$$\Phi_{\pi,i}=\left\{\widetilde{\Theta}_{\pi,i,\lambda}(\psi):\  \psi\in {\cal F}^{\pi,\lambda}_i, \,  \lambda\in\sigma(\pi(S))\right\}$$ is a frame for $\E_{\pi,i}$,
where $\sigma(\pi(S))$ denotes the spectrum of the matrix $\pi(S)$. 
\item[(ii)] The collection 
$\Phi=\bigcup_{\pi\in\widehat{\bbS_n}, 1\leq i\leq d_\pi}\Phi_{\pi,i}$
is a Frobenius--Schur frame for $L^2(\bbS_n)$, compatible with $S$. 
\item[(iii)] Every  frame for $L^2(\bbS_n)$  which is both Frobenius--Schur and compatible with $S$ is of the form described in (ii).
\end{itemize}
\end{theorem}
\begin{proof}
 First let $1\leq i\leq d_{\pi}$, and note that any frame for $E_{\lambda}(\pi(S))$ can be lifted, via the map $\widetilde{\Theta}_{\pi,i,\lambda}$, to a frame for $Z_{\pi, i, \lambda}$ with the same upper/lower frame bounds. That is, if $\{\phi_x\}_{x\in X}$ belongs to the frame space ${\mathcal G}_{\pi,\lambda}$, then 
 $\{\widetilde{\Theta}_{\pi,i,\lambda}(\phi_x)\}_{x\in X}$ forms a frame for $Z_{\pi, i, \lambda}$.
 This is indeed the case, as  $\widetilde{\Theta}_{\pi,i,\lambda}$ is an isometric isomorphism (Lemma~\ref{lemma: isometry}).
So $\Phi_{\pi,i}$ is just a union of frames for $Z_{\pi,i,\lambda}$, as $\lambda\in\sigma(\pi(S))$ varies.
Using the direct-sum decomposition of Theorem~\ref{thm:Z-pi-decomp}, this union results in a frame for $\E_{\pi,i}$. This proves (i).

We prove (ii) now. The fact that $\Phi$ is a frame for $L^2(\bbS_n)$ follows from part (i) and Theorem~\ref{thm:schur-ortho}, and  $\Phi$ being a compatible Frobenius--Schur frame follows directly from Lemma~\ref{lem:frob-schur-Z description}, Lemma~\ref{lemma: isometry} (ii) and the definition of $\Phi$.
Finally, part (iii) follows from the fact that a compatible Frobenius--Schur frame can be naturally partitioned into frames for $Z_{\pi,i,\lambda}$. Applying the map $\Theta^{-1}_{\pi,i,\lambda}$ to those frames finishes the proof.
\end{proof}

The construction in Theorem \ref{thm: frame} allows us to control the frame bounds as well. We get the following corollary.
\begin{corollary}
With notation as in Theorem~\ref{thm: frame}, suppose the frames ${\cal F}^{\pi,\lambda}_i$  are Parseval  (resp.~tight) for all $\pi$, 
$i$, and $\lambda$. 
Then $\Phi_{\pi,i}$ is a Parseval (resp.~tight) frame for $\E_{\pi, i}$, and $\Phi$ is a Parseval (resp.~tight) frame for $L^2({\mathbb S}_n)$.
\end{corollary}

%%%%%%%%%%%%%%%%%%%%%%%%%%%%%%%%%%%%%%%%%%%%%%%%%%%%%%%%%%%%%%%%%%%%%%%%%%%%%%%%%%%%%%%%%%%%%%%%%%%%%%%%%%%
\begin{remark}[Relation with Results in \cite{paper}]\label{rem:relation-to-paper}
Our frame construction relates with the dictionaries given in \cite{paper}. Let $\pi\in\widehat{\bbS_n}$ be associated with a partition $\gamma$ of $n$. Then the space $W_\gamma$ given in \cite[Equation (1)]{paper} is simply the subspace of ${\mathbb C}[\bbS_n]$ containing all coefficient functions associated with $\pi$, namely, 
$W_\gamma=\oplus_{i=1}^{d_\pi}{\cal E}_{\pi,i}.$
%
%Moreover, the space $U_{\lambda}$ given in \cite[Equation (6)]{paper} is the $\lambda$-eigenspace of $\pi(S)$. $Z_{\gamma,\lambda}=W_{\gamma}\cap U_{\lambda}$ 
One can verify that the space $Z_{\gamma,\lambda}$ as defined in \cite[Proposition 1]{paper} is exactly the space 
$Z_{\pi,\lambda}$ defined as 
$$Z_{\pi,\lambda}=\bigoplus_{i=1}^{d_{\pi}} Z_{\pi,i, \lambda}.$$
%reference \cite[Theorem~\ref{thm:eigenvector}(ii)]{2022:Beck:SPCayley} 
The frame constructed in \cite{paper} is therefore a Frobenius--Schur frame compatible with the specific generating set of all adjacent transpositions. 
\end{remark}

%%%%%%%%%%%%%%%%%%%%%%%%%%%%%%%%%%%%%%%%%%%%%%%%%%%%%%%%%%%%%%%%%%%%%%%%%%%%%%%%%%%%%%%%%%%%%%%%%%%%%%%%%%%
\section{Compatible Frobenius--Schur frames for $L^2(\mathbb S_3)$}\label{sec:example}
In this section, we provide a recipe to construct Frobenius--Schur frames for $L^2(\bbS_3)$ which are compatible with the generating set $S=\{(12),(23)\}$. 
A more interesting example would be the construction of Frobenius--Schur frames for $L^2(\bbS_n)$ with $n\geq 4$, because some of the associated eigenvalues will have higher multiplicity. Due to space constraints, we have chosen to only provide details for $L^2(\bbS_3)$; these steps, however, can be used as a guide when dealing with higher $n$.

We order elements of $\bbS_3$ as follows:  $\id,(12),(23),(13),(123),(132)$. From Subsection~\ref{subsec:rep}, we have three irreducible representations of $\bbS_3$, as there are only three Young diagrams of size 3 corresponding with the partitions (1,1,1), (2,1) and (3) respectively:
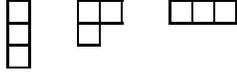
\begin{center}
\ytableausetup{smalltableaux}
\ydiagram{1,1,1}
\hspace{0.5cm}
\ydiagram{2,1}
\hspace{0.5cm}
\ydiagram{3}
\end{center}

For the partition (3), the only standard Young tableau is the rightmost Young diagram labeled with 1, 2, and 3 from left to right.
% \begin{center}
%     \begin{ytableau}
%     1 \\
%     2 \\
%     3
%     \end{ytableau}
% \end{center}
Therefore, the corresponding representation is 1-dimensional. The  representation associated with this partition is the {\bf trivial representation} of $\bbS_3$, denoted by $\iota$, which maps every element of $\bbS_3$
to 1. The unique coefficient function of $\iota$ is given by $\iota_{1,1}=[1,1,1,1,1,1]^{\t}.$

For the partition (2,1), the standard Young tableaux are
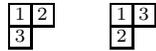
\begin{center}
\ytableausetup{smalltableaux}
    \begin{ytableau}
    1 & 2 \\
    3
    \end{ytableau}
    \hspace{0.5cm}
    \begin{ytableau}
    1 & 3 \\
    2
    \end{ytableau}
\end{center}
so the corresponding representation is 2-dimensional. The representation associated with this partition is the {\bf standard representation} of $\bbS_3$, denoted $\pi$ and defined as
follows:
\begin{equation*}
     \pi(\id)=I_2, \ 
    \pi(12)=\begin{bmatrix} -\frac{1}{2} & \frac{\sqrt{3}}{2}\\ \frac{\sqrt{3}}{2} & \frac{1}{2}\end{bmatrix},\
    \pi(23)=\begin{bmatrix} 1 & 0\\ 0 & -1\end{bmatrix}.
\end{equation*}
Since $\pi$ is multiplicative, and $\{(12),(23)\}$ is a generator for $\bbS_3$, the above matrices are enough to define $\pi$ on $\bbS_3$.
The coefficient functions of $\pi$ are
\begin{align*}
    &\pi_{1,1}=\left[1, \tfrac{-1}{2}, 1, \tfrac{-1}{2}, \tfrac{-1}{2},\tfrac{-1}{2}\right]^{\t},\\
    &\pi_{2,1}=\left[0, \tfrac{\sqrt{3}}{2}, 0, \tfrac{-\sqrt{3}}{2},\tfrac{-\sqrt{3}}{2},\tfrac{\sqrt{3}}{2}\right]^{\t},\\
    &\pi_{1,2}= \left[0, \tfrac{\sqrt{3}}{2}, 0, \tfrac{-\sqrt{3}}{2}, \tfrac{\sqrt{3}}{2}, \tfrac{-\sqrt{3}}{2}\right]^{\t},\\
    &\pi_{2,2}= \left[1, \tfrac{1}{2}, -1, \tfrac{1}{2}, \tfrac{-1}{2}, \tfrac{-1}{2}\right]^{\t}.
\end{align*}

For the partition (1,1,1), the only standard Young tableau is the leftmost Young diagram labeled with 1, 2, and 3 from top to bottom.
% \begin{center}
%     \begin{ytableau}
%     1 & 2 & 3
%     \end{ytableau}
% \end{center}
Then the corresponding representation is 1-dimensional. The representation associated with this partition is the {\bf alternating representation} of $\bbS_3$, denoted by $\tau$, which maps $\sigma\in \bbS_3$ to the sign of the permutation. The unique coefficient function of $\tau$ is given by $\tau_{1,1}=[1,-1,-1,-1,1,1]^{\t}.$
Next, we follow Theorem~\ref{thm: frame} to obtain all Frobenius--Schur frames compatible with $S$. Note that 
%$$\pi(S)=\pi(12)+\pi(23)=\begin{bmatrix} \frac{1}{2} & \frac{\sqrt{3}}{2} \\ \frac{\sqrt{3}}{2} & -\frac{1}{2}\end{bmatrix},$$ 
$\pi(S)=\pi(12)+\pi(23)$ has  eigenvalues  $-1$ and $1$, and the corresponding eigenvectors are 
$v_1=[\sqrt{3},  1]^{\t}, \ v_{-1}=[-\frac{1}{\sqrt{3}},  1]^{\t}.$
%$$v_1=\begin{bmatrix} \sqrt{3}\\ 1 \end{bmatrix} \text{ and } v_{-1}=\begin{bmatrix} -\frac{1}{\sqrt{3}} \\ 1 \end{bmatrix}.$$
%
Taking $i=1,2$ and $\lambda=1,-1$, we get four $Z_{\pi, i, \lambda}$ spaces as defined in Theorem~\ref{thm: frame},
\begin{align*}
&Z_{\pi, 1, 1}=\bbC\left[\sqrt{3},0,\sqrt{3},-\sqrt{3},-\sqrt{3},0\right]^{\t}\\
&Z_{\pi, 1, -1}=\bbC\left[\tfrac{-1}{\sqrt{3}},\tfrac{2}{\sqrt{3}},\tfrac{-1}{\sqrt{3}},\tfrac{-1}{\sqrt{3}},\tfrac{-1}{\sqrt{3}},\tfrac{2}{\sqrt{3}}\right]^{\t}\\
&Z_{\pi, 2, 1}=\bbC\left[1,2,-1,-1,1,-2\right]^{\t}\\
&Z_{\pi, 2, -1}=\bbC\left[1,0,-1,1,-1,0\right]^{\t}.
\end{align*}

%Z_{\pi, 1, 1}&=&\{c\sqrt{3}\pi_{1,1}+c\pi_{2,1}: c\in \bbC\}\\
%Z_{\pi, 1, -1}&=&\{c\frac{-1}{\sqrt{3}}\pi_{1,1}+c\pi_{2,1}: c\in \bbC\}\\
%Z_{\pi, 2, 1}&=&\{c\sqrt{3}\pi_{1,2}+c\pi_{2,2}: c\in \bbC\}\\
%Z_{\pi, 2, -1}&=&\{c\frac{-1}{\sqrt{3}}\pi_{1,2}+c\pi_{2,2}: c\in \bbC\}\\

For the 1-dimensional representation $\iota$, we have $\iota(S)=\iota(12)+\iota(23)=2$, so $\lambda=2$. Then we get the space $$Z_{\iota, 1, 2}=\bbC\left[1,1,1,1,1,1\right]^{\t}.$$

%For the 1-dimensional representation $\tau$, we have $\tau(S)=\tau(12)+\tau(23)=-2$, so $\lambda=-2$. Then 
Similarly, we get the space 
$$Z_{\tau, 1, -2}=\bbC\left[1,-1,-1,-1,1,1\right]^{\t}.$$

The Frobenius--Schur frames compatible with $S$ are precisely the union of frames for  $Z_{\iota, 1, 2}, Z_{\tau, 1, -2}$, $Z_{\pi,i,\pm1}$ (for $i=1,2$).

% \textcolor{blue}{S4 and beyond is the interesting one, because we have eigenvalues with multiplicities. But no spaces here.}
% \textcolor{blue}{To save space, please make the pictures smaller, and replace the full first names in the bib with initials. (up to date)}

%\section*{Acknowledgment}
%thank referee.

%%%%%%%%%%%%%%%%%%%%%%%%%%%%%%%%%%%%%%%%%%%%%%%%%%%%%%%%%%%%%%%%%%%%%%%%%%%%%%%%%%%%%%%%%%%%%%%%%%%%%%%%%%%%%%%%%%%%%%%%%%%%%%%%%%%%%%%%%%%%%%%%%%%%%%%%
%BIBLIOGRAPHY

%\bibliographystyle{plain}
%\bibliography{refs}

\begin{thebibliography}{10}

\bibitem{2016:Behjat:SignalAdaptedFrames}
H.~Behjat, U.~Richter, D.~Van De~Ville, and L.~Sörnmo.
\newblock Signal-adapted tight frames on graphs.
\newblock {\em IEEE Transactions on Signal Processing}, 64(22):6017--6029,
  2016.

\bibitem{paper}
Y.~Chen, J.~DeJong, T.~Halverson, and D.~I. Shuman.
\newblock Signal {P}rocessing on the {P}ermutahedron: {T}ight {S}pectral
  {F}rames for {R}anked {D}ata {A}nalysis.
\newblock {\em J. Fourier Anal. Appl.}, 27(4):Paper No. 70, 2021.

\bibitem{2017:Dong:SparseRepWavelet}
B.~Dong.
\newblock Sparse representation on graphs by tight wavelet frames and
  applications.
\newblock {\em Appl. Comput. Harmon. Anal.}, 42(3):452--479, 2017.

\bibitem{FultonHarris}
W.~Fulton and J.~Harris.
\newblock {\em Representation theory. A first course}, volume 129 of {\em
  Graduates Texts in Mathematics}.
\newblock Springer-Verlag, New York, first edition, 1991.

\bibitem{2021:Ghandehari:GaborTypeFrames}
M.~Ghandehari, D.~Guillot, and K.~Hollingsworth.
\newblock Gabor-type frames for signal processing on graphs.
\newblock {\em J. Fourier Anal. Appl.}, 27(2):Paper No. 25, 23, 2021.

\bibitem{2018:Gobel:TightFramesDenoising}
F.~G\"{o}bel, G.~Blanchard, and U.~von Luxburg.
\newblock Construction of tight frames on graphs and application to denoising.
\newblock In {\em Handbook of big data analytics}, Springer Handb. Comput.
  Stat., pages 503--522. Springer, Cham, 2018.

\bibitem{Hammond:2011:WaveletsOnGraphViaSpectralGraphTheory}
D.~K. Hammond, P.~Vandergheynst, and R.~Gribonval.
\newblock Wavelets on graphs via spectral graph theory.
\newblock {\em Applied and Computational Harmonic Analysis}, 30(2):129--150,
  2011.

\bibitem{2013:Leonardi:TightWaveletFrames}
N.~Leonardi and D.~Van De~Ville.
\newblock Tight wavelet frames on multislice graphs.
\newblock {\em IEEE Trans. Signal Process.}, 61(13):3357--3367, 2013.

\bibitem{2018:Ortega:GraphSignalOverview}
A.~Ortega, P.~Frossard, J.~Kovačević, J.~M.~F. Moura, and P.~Vandergheynst.
\newblock Graph signal processing: Overview, challenges, and applications.
\newblock {\em Proceedings of the IEEE}, 106(5):808--828, 2018.

\bibitem{Serre}
J-P. Serre.
\newblock {\em Linear representations of finite groups}.
\newblock Springer-Verlag, New York-Heidelberg, 1977.
\newblock Translated from the second French edition by Leonard L. Scott,
  Graduate Texts in Mathematics, Vol. 42.

\bibitem{2013:Shuman:SpectrumAdaptedWavelet}
D.~Shuman, C.~Wiesmeyr, N.~Holighaus, and P.~Vandergheynst.
\newblock Spectrum-adapted tight graph wavelet and vertex-frequency frames.
\newblock {\em IEEE Transactions on Signal Processing}, 63, 11 2013.

\bibitem{2020:Shuman:UnifyingFramework}
D.~I. Shuman.
\newblock Localized spectral graph filter frames: A unifying framework, survey
  of design considerations, and numerical comparison.
\newblock {\em IEEE Signal Processing Magazine}, 37(6):43--63, 2020.

\bibitem{2013:Shuman:EmergingFieldSignal}
D.~I. Shuman, S.~K. Narang, P.~Frossard, A.~Ortega, and P.~Vandergheynst.
\newblock The emerging field of signal processing on graphs: Extending
  high-dimensional data analysis to networks and other irregular domains.
\newblock {\em IEEE Signal Processing Magazine}, 30(3):83--98, 2013.

\bibitem{2019:Stankovic:IntroGSP}
L.~Stankovi\'{c}, M.~Dakovi\'{c}, and E.~Sejdi\'{c}.
\newblock Introduction to graph signal processing.
\newblock In {\em Vertex-frequency analysis of graph signals}, Signals Commun.
  Technol., pages 3--108. Springer, Cham, 2019.

\end{thebibliography}

\end{document}